\documentclass[12pt,reqno,a4paper]{amsart}
\usepackage{amsmath,amssymb,dsfont,graphicx,cite,latexsym,
epsf,cancel,epic,eepic}
\usepackage{amssymb,mathrsfs}
\usepackage[colorlinks=false]{hyperref}
\usepackage{mathpazo}
\newtheorem{theorem}{Theorem}
\newtheorem{proposition}[theorem]{Proposition}

\newtheorem{lemma}[theorem]{Lemma}

\newtheorem{cor}[theorem]{Corollary}
\def\beq{\begin{equation}}
\def\eeq{\end{equation}}
\def\bea{\begin{eqnarray}}
\def\eea{\end{eqnarray}}
\def\benpf{\noindent {\textbf{{\emph{Proof.}}\;}}}
\def\endpf{\hfill$\blacksquare$\medskip}
\makeatletter
\expandafter\let\expandafter
\reset@font\csname reset@font\endcsname
\def\subeqnarray{\arraycolsep1pt
   \def\@eqnnum\stepcounter##1{\stepcounter{subequation}
       {\reset@font\rm(\theequation\alph{subequation})}}
\jot5mm     \eqnarray}

\makeatother
\newcommand{\bbR}{{\mathbb R}}

\def\ep{\varepsilon}
\def\epsilon{\varepsilon}
\def\t{\widetilde}
\def\pa{\partial}

\def\endpf{\hfill$\square$\medskip}

\newbox\meibox
\def\placeunder#1#2#3#4{\setbox\meibox%
\vbox{\hbox{\hskip#4$\hphantom{#2}$}\hbox{$\hphantom{#1}$}}%
\vtop{\baselineskip=0pt\lineskiplimit=\baselineskip%
\lineskip=#3\hbox to \wd\meibox{\hfil\hskip#4$#2$\hfil}%
\hbox to \wd\meibox{\hfil$#1$\hfil}}}
\def\intprod{\mathbin{\hbox to 6pt{%
                 \vrule height0.4pt width5pt depth0pt
                 \kern-.4pt
                 \vrule height6pt width0.4pt depth0pt\hss}}}

\begin{document}
\title[Commuting pairs of integrable symplectic birational maps of $\mathbb R^4$]
{A construction of a large family of commuting pairs \\ of integrable symplectic birational 4-dimensional maps}

\author{Matteo Petrera \and Yuri B. Suris }

\thanks{E-mail: {\tt  petrera@math.tu-berlin.de, suris@math.tu-berlin.de}}

\maketitle

\begin{center}
{\footnotesize{
Institut f\"ur Mathematik, MA 7-2\\
Technische Universit\"at Berlin, Str. des 17. Juni 136,
10623 Berlin, Germany
}}
\end{center}


\begin{abstract}
We give a construction of completely integrable 4-dimensional Hamiltonian systems with cubic Hamilton functions. 
Applying to the corresponding pairs of commuting quadratic Hamiltonian vector fields the so called Kahan-Hirota-Kimura discretization scheme, we arrive at pairs of birational 4-dimensional maps. We show that these maps are symplectic with respect to a symplectic structure that is a perturbation of the standard symplectic structure on $\mathbb R^4$, and possess two independent integrals of motion, which are perturbations of the original Hamilton functions. Thus, these maps are completely integrable in the Liouville-Arnold sense. Moreover, under a suitable normalization of the original pairs of vector fields, the pairs of maps commute and share the invariant symplectic structure and the two integrals of motion.
\end{abstract}

\section{Introduction}
\label{sect intro}

The theory of integrable systems has a long record of fruitful interactions with various branches of mathematics, most prominently with algebraic geometry. A bright example is an in-depth study of the geometry of elliptic rational surfaces and their automorphisms by Duistermaat \cite{Dui}, following the discovery of a large QRT family of integrable birational 2-dimensional maps by Quispel, Roberts and Thompson \cite{QRT}.  

The goal of this paper is to introduce a large family of integrable 4-dimensional maps, along with their plenty of remarkable properties. We hope that the study of the algebraic geometry of these maps will turn out to be as fruitful as the one in \cite{Dui}. This family comes as a new instance in a long and still mysterious list of features of the so called Kahan-Hirota-Kimura discretization method for quadratic vector fields.

This method was introduced in the geometric integration literature by Kahan in the unpublished notes \cite{K} as a method applicable to any system of ordinary differential equations for $x:\bbR\to\bbR^n$ with a quadratic vector field:
\begin{equation}\label{eq: diff eq gen}
\dot{x}=f(x)=Q(x)+Bx+c,
\end{equation}
where each component of $Q:\bbR^n\to\bbR^n$ is a quadratic form, while $B\in{\rm Mat}_{n\times n}(\bbR)$ and $c\in\bbR^n$. Kahan's discretization (with stepsize $2\epsilon$) reads as
\begin{equation}\label{eq: Kahan gen}
\frac{\widetilde{x}-x}{2\epsilon}=Q(x,\widetilde{x})+\frac{1}{2}B(x+\widetilde{x})+c,
\end{equation}
where
\[
Q(x,\widetilde{x})=\frac{1}{2}\big(Q(x+\widetilde{x})-Q(x)-Q(\widetilde{x})\big)
\]
is the symmetric bilinear form corresponding to the quadratic form $Q$. Equation (\ref{eq: Kahan gen}) is {\em linear} with respect to $\widetilde x$ and therefore defines a {\em rational} map $\widetilde{x}=\Phi_f(x,\epsilon)$. Clearly, this map approximates the time $2\epsilon$ shift along the solutions of the original differential system. Since equation (\ref{eq: Kahan gen}) remains invariant under the interchange $x\leftrightarrow\widetilde{x}$ with the simultaneous sign inversion $\epsilon\mapsto-\epsilon$, one has the {\em reversibility} property
\begin{equation}\label{eq: reversible}
\Phi_f^{-1}(x,\epsilon)=\Phi_f(x,-\epsilon).
\end{equation}
In particular, the map $f$ is {\em birational}. 
Kahan applied this discretization scheme to the famous Lotka-Volterra system and showed that in this case it possesses a very remarkable non-spiralling property. This property was explained by Sanz-Serna \cite{SS} by demonstrating that in this case the numerical method preserves an invariant Poisson structure of the original system.

The next intriguing appearance of this discretization was in the two papers by Hirota and Kimura who (being apparently unaware of the work by Kahan) applied it to two famous {\em integrable} system of classical mechanics, the Euler top and the Lagrange top \cite{HK, KH}. Surprisingly, the discretization scheme produced in both cases {\em integrable} maps. 

In \cite{PS, PPS1, PPS2} the authors undertook an extensive study of the properties of the Kahan's method when applied to integrable systems (we proposed to use in the integrable context the term ``Hirota-Kimura method''). It was demonstrated that, in an amazing number of cases, the method preserves integrability in the sense that the map $\Phi_f(x,\epsilon)$ possesses as many independent integrals of motion as the original system $\dot x=f(x)$.

Further remarkable geometric properties of the Kahan's method were discovered by Celledoni, McLachlan, Owren and Quispel in \cite{CMOQ1}, see also \cite{CMOQ2, CMOQ3}.  They demonstrated that for an arbitrary Hamiltonian vector field $f(x)=J\nabla H(x)$ with a constant Posson tensor $J$ and a cubic Hamilton function $H(x)$ the map $\Phi_f(x,\epsilon)$ possesses a rational integral of motion $\t H(x,\epsilon)$ such that $\t H(x,0)=H(x)$, and an invariant measure with a rational density, which is a small perturbation of the phase volume $dx_1\wedge \ldots \wedge dx_n$ as $\epsilon\to 0$. It should be mentioned that, while for $n=2$ the existence of an invariant measure is equivalent to symplecticity of the map $\Phi_f(x,\epsilon)$, the latter property was not established for any quadratic Hamiltonian system  in dimension $n\ge 4$.

In the present paper we give a construction of a big family of completely integrable Hamiltonian systems in dimension $n=4$ for which the Kahan-Hirota-Kimura discretization possesses a whole series of novel features. 

The main set of parameters of the construction is encoded in a $4\times 4$ matrix 
$$
A=\begin{pmatrix}
a_1 & a_2 & 0 &-a_5 \\
a_3 & -a_1 & a_5 & 0 \\
0& a_6 & a_1 &a_3\\
-a_6 &0 & a_2 & -a_1
\end{pmatrix}.
$$
Such matrices form a {\em 5-dimensional} vector space. To each non-degenerate matrix from this space, there corresponds a {\em 8-dimensional} vector space of cubic polynomials $H(x)$ on $\mathbb R^4$, satisfying a certain system of second order linear PDEs, encoded in the matrix equation 
\beq \label{harm}
A(\nabla^2 H)=(\nabla^2 H) A^{\rm T},
\end{equation}
 where $\nabla^2 H$ is the Hesse matrix of the function $H$. To each such polynomial  $H(x)$ there corresponds a {\em unique} ``dual'' polynomial $K(x)$ from the same 8-dimensional vector space, characterized by 
\begin{equation}\label{CR}
\nabla K(x)=A\nabla H(x).
\end{equation}
One can think of equation \eqref{harm} as a matrix analog of the Laplace equation for harmonic functions, and of equation \eqref{CR} as an analog of the Cauchy-Riemann equations relating conjugate pairs of harmonic functions.
It turns out that the functions $H(x)$ and $K(x)$ are functionally independent and are in involution with respect to the standard symplectic structure on $\mathbb R^4$. Therefore, each one defines a completely integrable Hamiltonian system. The flows of the Hamiltonian vector fields $J\nabla H(x)$ and $J\nabla K(x)$ commute. The following are the striking properties of the corresponding Kahan-Hirota-Kimura discretizations. 
\begin{itemize}
\item The map $\Phi_{J\nabla H}$ is symplectic with respect to a symplectic structure which is a perturbation of the canonical symplectic structure on $\mathbb R^4$, and possesses two functionally independent integrals. In other words, $\Phi_{J\nabla H}$ is completely integrable. Of course, the same holds true for $\Phi_{J\nabla K}$.
\item The integrals $\t H(x,\epsilon)$ and $\t K(x,\epsilon)$ of $\Phi_{J\nabla H}$ are rational perturbations of the original polynomials $H(x)$ and $K(x)$, are related by the same equation as \eqref{CR}, that is, $\nabla \t K=A \nabla\t H$, and satisfy the same second order differential equations \eqref{harm} as $H(x)$ and $K(x)$ do.
\item There exists a unique (up to sign) number $\alpha$ such that the maps $\Phi_{J\nabla H}$ and $\Phi_{\alpha^{-1}J\nabla K}$ commute. For this value of $\alpha$, the maps $\Phi_{J\nabla H}$ and $\Phi_{\alpha^{-1}J\nabla K}$ share the invariant symplectic structure and the two functionally independent integrals.
\end{itemize}

We provide the reader with a quick reminder about the general properties of the Kahan-Hirota-Kimura discretization method in Sect. \ref{sect HK}.
Then we discuss the details of the general construction of the dual pairs $H(x)$, $K(x)$ in Sect. \ref{sect int fam}. The rich algebraic properties of the corresponding vector fields $J\nabla H$ and $J\nabla K$ are collected in Sect. \ref{sect vector fields}. On the basis of these properties, we prove the main results in Sect. \ref{sect commut} (commutativity), Sect. \ref{sect integrals} (two integrals of motion), Sect. \ref{sect symplectic} (invariant symplectic structure) and Sect. \ref{sect diff eqs} (differential equations for the integrals of the maps).

\section{General properties of the Kahan-Hirota-Kimura discretization}
\label{sect HK}

Here we recall the main properties of the Kahan-Hirota-Kimura discretization, following mainly \cite{PPS1, PPS2} and \cite{CMOQ1}. 

The explicit form of the map $\Phi_f$ defined by \eqref{eq: Kahan gen} is 
\beq \label{eq: Phi gen}
\t x =\Phi_f(x,\ep)= x + 2\ep \left( I - \ep f'(x) \right)^{-1} f(x),
\eeq
where $f'(x)$ denotes the Jacobi matrix of $f(x)$. Moreover, if the vector field $f(x)$ is homogeneous (of degree 2), then \eqref{eq: Phi gen} can be equivalently rewritten as
\beq \label{eq: Phi hom}
\t x =\Phi_f(x,\ep)= \left( I - \epsilon f'(x) \right)^{-1} x.
\eeq
Due to \eqref{eq: reversible}, in the latter case we also have:
\beq \label{eq: Phi hom alt}
x =\Phi_f(\t x,-\ep)= \left( I + \epsilon f'(\t x) \right)^{-1} \t x \quad \Leftrightarrow \quad \t x = \left( I + \epsilon f'(\t x) \right) x.
\eeq
One has the following expression for the Jacobi matrix of the map $\Phi_f$:
\beq \label{Jac}
d\Phi_f(x)=\frac{\partial\t x}{\partial x}=\big(I-\epsilon f'(x)\big)^{-1}\big(I+\epsilon f'(\t x)\big).
\eeq
Let the vector field $f(x)$ be Hamiltonian, $f(x)=J\nabla H(x)$, where $H:\bbR^n \to\bbR$ is a cubic polynomial and $J$ is a non-degenerate skew-symmetric $n\times n$ matrix, so that $H(x)$ is an integral of motion for $\dot x=f(x)$. Then the map $\Phi_f(x,\epsilon)$ possesses the following rational integral of motion: 
\beq\label{eq: CMOQ integral}
\t H(x,\epsilon) = \frac{1}{6\epsilon}x^{\rm T}J^{-1}\t x=H(x)+\frac{2\ep}{3}  (\nabla H(x))^{\rm T} \left( I - \ep f'(x) \right)^{-1} f(x),
\eeq
as well as an invariant measure
\beq\label{eq: CMOQ measure}
\frac{dx_1\wedge\ldots\wedge dx_n}{\displaystyle{\det\left( I - \ep f'(x) \right)}}.
\eeq
These remarkable results from \cite{CMOQ1} hold true for all quadratic Hamiltonian vector fields and therefore are not related to integrability.

\section{A family of integrable 4-dimensional Hamiltonian systems}
\label{sect int fam}


Consider the canonical phase space $\bbR^{4}$ with coordinates $(x_1,x_{2},x_{3},x_{4})$, equipped with the standard symplectic structure in which the Poisson brackets of the coordinate functions are $\{x_1,x_3\}= \{x_2,x_4\}=1$ (all other brackets being either obtained from these ones by skew-symmetry or otherwise vanish). Let $H(x)=H(x_1,x_{2},x_{3},x_{4})$  be a Hamilton function on  $\bbR^{4}$. The corresponding Hamiltonian system is governed by the equations of motion
\beq  \label{Ham}
\dot x= J \nabla H(x),\quad J=\begin{pmatrix} 0 & I \\ -I & 0 \end{pmatrix}.
\eeq
\begin{proposition}
Consider a constant non-degenerate $4\times 4$ matrix $A$, and suppose that for a function $K(x)=K(x_1,x_{2},x_{3},x_{4})$ the following relations are satisfied:
\beq \label{grad K}
\nabla K = A \nabla H.
\eeq
If the matrix $A$ satisfies
\beq \label{cond A}
 A^{\rm{T}} J = J A,
\eeq
then the functions $H$, $K$ are in involution, so that the Hamiltonian system \eqref{Ham} is completely integrable.
\end{proposition}

\benpf We have:
$$
\{H,K\} = (\nabla H)^{\rm{T}} J \nabla K=(\nabla H)^{\rm{T}} J A \nabla H,
$$
and this vanishes if the matrix $JA$ is skew-symmetric, which gives condition \eqref{cond A}. Equation \eqref{grad K} with a non-scalar matrix $A$ also ensures that $H$, $K$ are functionally independent (generically).
\endpf

If $A$ is written in the block form as
$$
A = \begin{pmatrix}
A_1 & A_2 \\
A_3 & A_4
\end{pmatrix}
$$
with $2\times 2$ blocks $A_i$, then condition \eqref{cond A} reads:
\beq \label{cond2}
A_1^{\rm{T}}=A_4, \qquad A_2^{\rm{T}}=-A_2, \qquad A_3^{\rm{T}}=-A_3.
\eeq
Such matrices form a six-dimensional vector space:
\beq \label{A prelim}
A=\begin{pmatrix}
a_1 & a_2 & 0 &-a_5 \\
a_3 & a_4 & a_5 & 0 \\
0& a_6 & a_1 &a_3\\
-a_6 &0 & a_2 & a_4
\end{pmatrix}.
\eeq

We now discuss applicability of this construction. For a given function $H$, differential equations \eqref{grad K} for $K$ are solvable if and only if $H$ satisfies the following condition:
\beq  \label{d2H}
A(\nabla^2 H)=(\nabla^2 H) A^{\rm T},
\eeq
where $\nabla^2 H$ is the Hesse matrix of the function $H$. Written down explicitly, one has the following five second order PDEs for $H$:
\bea
&& C_1=-a_3 \frac{\pa^2 H}{\pa x_1^2}
+a_2 \frac{\pa^2 H}{\pa x_2^2}
+(a_1-a_4) \frac{\pa^2 H}{\pa x_1 \pa x_2}
-a_5 \left(\frac{\pa^2 H}{\pa x_1 \pa x_3}
+ \frac{\pa^2 H}{\pa x_2 \pa x_4}\right) =0, \label{comp1} \\
&& C_2= -a_2 \frac{\pa^2 H}{\pa x_3^2}
+a_3 \frac{\pa^2 H}{\pa x_4^2}
+(a_1-a_4) \frac{\pa^2 H}{\pa x_3 \pa x_4}
+a_6\left( \frac{\pa^2 H}{\pa x_2 \pa x_4}
+ \frac{\pa^2 H}{\pa x_1 \pa x_3}\right)=0, \label{comp2} \\
&& C_3=a_6 \frac{\pa^2 H}{\pa x_1^2}
-a_5 \frac{\pa^2 H}{\pa x_4^2}
+(a_1-a_4) \frac{\pa^2 H}{\pa x_1 \pa x_4}
+a_2 \left(
 \frac{\pa^2 H}{\pa x_2 \pa x_4}
 -\frac{\pa^2 H}{\pa x_1 \pa x_3}\right)=0, \label{comp3} \\
&& C_4=  a_5 \frac{\pa^2 H}{\pa x_3^2}
-a_6 \frac{\pa^2 H}{\pa x_2^2}
-(a_1-a_4) \frac{\pa^2 H}{\pa x_2 \pa x_3}
+a_3 \left( \frac{\pa^2 H}{\pa x_1 \pa x_3}
- \frac{\pa^2 H}{\pa x_2 \pa x_4}\right)=0, \label{comp4} \\
&& C_5=-a_6 \frac{\pa^2 H}{\pa x_1 \pa x_2}
-a_5 \frac{\pa^2 H}{\pa x_3 \pa x_4}
+a_2 \frac{\pa^2 H}{\pa x_2 \pa x_3}
-a_3 \frac{\pa^2 H}{\pa x_1 \pa x_4}
=0.                                                             \label{comp5} 
\eea
Only four of these PDEs are linearly independent, due to the linear relation
$$
a_6 C_1+a_5 C_2+a_3 C_3+a_2 C_4+(a_1-a_4)C_5=0.
$$
\begin{proposition}
The linear space of homogeneous polynomials of degree 3 satisfying the system of second order PDEs (\ref{comp1})--(\ref{comp5}), has dimension 8.
\end{proposition}
\begin{proof}
A general homogeneous polynomial of degree 3 in four variables $x_1$, $x_2$, $x_3$, $x_4$ has 20 coefficients. Each of the expressions
$C_i$, $i=1,\ldots 4$, is a linear polynomial in these four variables, so that each equation $C_i=0$ results in four linear equations for the coefficients of $H$. Altogether we get 16 linear homogeneous equations for 20 coefficients of $H$. A careful inspection of the resulting linear system reveals that the rank of its matrix is equal to 12. Therefore, the dimension of the space of solutions is equal to $20-12=8$. 
\end{proof}
Note that the solutions $K(x)$ of the first order PDEs (\ref{grad K}) satisfy that the same compatibility conditions (\ref{comp1})--(\ref{comp5}). To see this, observe that one has $\nabla H=A^{-1} \nabla K$ and 
\beq \label{A inv}
A^{-1}=\frac{1}{D}\begin{pmatrix}  -a_4 & a_2 & 0 & -a_5 \\
                                                              a_3 & -a_1 & a_5 & 0\\
                                                              0 & a_6 & -a_4 & a_3\\
                                                               -a_6 & 0 & a_2 & -a_1
                                   \end{pmatrix} , \quad {\rm where}\quad D=-a_1a_4+a_2a_3+a_5a_6\neq 0.
\eeq
Clearly, the set of PDEs (\ref{comp1})--(\ref{comp5}) generated by the latter matrix coincides with the original one. Thus, to any solution $H$ of the system  (\ref{comp1})--(\ref{comp5}) there corresponds, via \eqref{grad K}, another solution $K$ (unique up to an additive constant). Changing $A$ to $A+\beta I$ would lead to changing $K$ to $K+\beta H$, and would not change the set of all linear combinations of $H$ and $K$. We will use this freedom to ensure that 
\beq \label{tr=0}
{\rm tr}\ A=0 \quad \Leftrightarrow \quad a_4=-a_1.
\eeq
Thus, from now on we assume that
\beq \label{A red}
A=\begin{pmatrix}
a_1 & a_2 & 0 &-a_5 \\
a_3 & -a_1 & a_5 & 0 \\
0& a_6 & a_1 &a_3\\
-a_6 &0 & a_2 & -a_1
\end{pmatrix} .
\eeq
We mention also the following properties of the matrices $A$ as in \eqref{A red}:
\beq \label{A square}
A^{-1}=\frac{1}{D} A \; \Leftrightarrow\; A^2=DI, \quad {\rm where}\quad D=a_1^2+a_2a_3+a_5a_6.
\eeq
This follows immediately from \eqref{A inv}. Furthermore, a direct computation shows that
\beq \label{A ev}
\det(A-\lambda I)=(\lambda^2-D)^2.
\eeq

\section{General algebraic properties of the vector fields $f$, $g$}
\label{sect vector fields}

From now on we will assume that $A$ is a matrix as in \eqref{A red}. Define $D$ is as in \eqref{A square}, and let $\alpha$ be a real or a purely imaginary number (depending on whether $D>0$ or $D<0$) satisfying 
\beq \label{alpha}
\alpha^2=D.
\eeq 
Assume that $H(x)$ and $K(x)$ are homogeneous polynomials of degree 3 satisfying \eqref{grad K}. Set 
\beq \label{fg}
f(x)=J\nabla H(x), \quad g(x)=\alpha^{-1}J\nabla K(x)=\alpha^{-1}JA\nabla H(x).
\eeq
Due to \eqref{cond A}, we have:
\beq \label{f vs g}
g(x)=\alpha^{-1}A^{\rm T} f(x)\quad \Leftrightarrow \quad f(x)=\alpha^{-1}A^{\rm T}g(x).
\eeq
This means that the roles of the vector fields $f$, $g$ in all our constructions are absolutely symmetric. Furthermore,
\beq \label{f'g'}
f'(x)=J\nabla^2 H(x), \quad g'(x)=\alpha^{-1} J\nabla^2 K(x)=\alpha^{-1}JA\nabla^2 H(x).
\eeq
Again, due to \eqref{cond A}, we have:
\beq \label{f' vs g'}
g'(x)=\alpha^{-1}A^{\rm T} f'(x)\quad \Leftrightarrow \quad f'(x)=\alpha^{-1}A^{\rm T}g'(x).
\eeq
\begin{lemma}
The following identities hold true:
\beq \label{Ham f'}
(f'(x))^{\rm T}J=-Jf'(x), \quad (g'(x))^{\rm T}J=-Jg'(x),
\eeq
\beq \label{Af'}
 A^{\rm T}f'(x)=f'(x)A^{\rm T}, \quad  A^{\rm T}g'(x)=g'(x)A^{\rm T},
\eeq
\beq \label{vf comm}
 f'(x)g(x)=g'(x)f(x),
\eeq
\beq \label{f'f=g'g}
 f'(x)f(x)=g'(x)g(x),
\eeq
\beq \label{f'g'=g'f'}
f'(x)g'(x)=g'(x)f'(x),
\eeq
\beq \label{f'f'=g'g'}
(f'(x))^2=(g'(x))^2.
\eeq
\end{lemma}
\begin{proof}
Equation \eqref{Ham f'} is the characteristic property of Jacobi matrices of Hamiltonian vector fields.
Equation \eqref{Af'} is equivalent to \eqref{d2H}, due to \eqref{cond A}.
To prove \eqref{vf comm}, \eqref{f'f=g'g}, we compute with the help of \eqref{f vs g}, \eqref{Af'}:
\[
f'(x)g(x)=\alpha^{-1}A^{\rm T}g'(x)g(x)=\alpha^{-1}g'(x)A^{\rm T}g(x)=g'(x)f(x),
\]
and similarly,
\[
f'(x)f(x)=\alpha^{-1}A^{\rm T}g'(x)f(x)=\alpha^{-1}g'(x)A^{\rm T}f(x)=g'(x)g(x).
\]
Observe that \eqref{vf comm} expresses commutativity of the vector fields $f(x)$ and $g(x)$. 
Identities \eqref{f'g'=g'f'}, \eqref{f'f'=g'g'} are proved along the same lines, with the help of \eqref{f' vs g'}, \eqref{Af'}:
\[
f'(x)g'(x)=\alpha^{-1}A^{\rm T}g'(x)g'(x)=\alpha^{-1}g'(x)A^{\rm T}g'(x)=g'(x)f'(x),
\]
and
\[
(f'(x))^2=\alpha^{-1}A^{\rm T}g'(x)f'(x)=\alpha^{-1}g'(x)A^{\rm T}f'(x)=(g'(x))^2.
\]
\end{proof}
\begin{cor}\label{th1}
There holds:
\beq \label{condet}
\det  \left(I - \ep f'(x) \right)= \det  \left(I - \ep g'(x) \right) .
\eeq
\end{cor}
\begin{proof}
It is enough to prove that ${\rm tr}\, (f'(x))^k={\rm tr}\, (g'(x))^k$ for $k=1,2,3,4$. Since vector fields $f,g$ are Hamiltonian, we have:
${\rm tr}\, (f'(x))^k=0$ and ${\rm tr}\, (g'(x))^k=0$ for $k=1,3$. The equalities for $k=2,4$ follow from \eqref{f'f'=g'g'}.
\end{proof}

\begin{lemma}
The following identities hold true:
\beq \label{det(f+g)}
\det\big(f'(x)+g'(x)\big)=0, \quad \det\big(f'(x)-g'(x)\big)=0.
\eeq
\end{lemma}
\begin{proof}
Indeed, these determinants are equal to 
$$
\det\big((I\pm\alpha^{-1}A^{\rm T})f'(x)\big),
$$
and they both vanish due to $\det(\alpha I\pm A)=0$, which is a direct consequence of \eqref{A ev}, \eqref{alpha}.
\end{proof}

\begin{lemma} \label{lemma f'f'}
The following identities hold true:
\beq
(f'(x))^2=(g'(x))^2=p(x)I+q(x)A^{\rm T},
\eeq
where
\begin{eqnarray}
p(x)  & = & \frac{1}{4}{\rm tr} (f'(x))^2,  \label{p}\\
q(x) & = & \frac{1}{4D}{\rm tr}\big(A^{\rm T}(f'(x))^2\big)=\frac{1}{4\alpha}{\rm tr}\big(f'(x)g'(x)\big). \label{q}
\end{eqnarray}
\end{lemma}
\begin{proof}
We have, due to \eqref{det(f+g)}:
\begin{eqnarray}
\det(\lambda I- (f'-g')) & = & \lambda^4-\frac{1}{2}\lambda^2 \ {\rm tr} (f'-g')^2, \label{char pol f'-g'}\\
\det(\lambda I- (f'+g')) & = & \lambda^4-\frac{1}{2}\lambda^2 \ {\rm tr} (f'+g')^2. \label{char pol f'+g'}
\end{eqnarray}
By the theorem of Caley-Hamilton, we have:
\begin{eqnarray}
(f'-g')^4-\frac{1}{2}(f'-g')^2 \ {\rm tr} (f'-g')^2 & = & 0, \label{pol f'-g'}\\
(f'+g')^4-\frac{1}{2}(f'+g')^2 \ {\rm tr} (f'+g')^2 & = & 0. \label{pol f'+g'}
\end{eqnarray}
Add the latter two identities, taking into account \eqref{f'g'=g'f'} and \eqref{f'f'=g'g'}. The result reads:
$$
16 (f')^4 -4(f')^2 \ {\rm tr} (f')^2-4f'g' \ {\rm tr}(f'g')=0,
$$
or, equivalently,
$$
(f')^4-\frac{1}{4}(f')^2 \ {\rm tr} (f')^2 -\frac{1}{4D}A^{\rm T}(f')^2 \ {\rm tr} (A^{\rm T} (f')^2)=0.
$$
Upon dividing by a generically non-degenerate matrix $(f')^2$ (its determinant is a non-vanishing homogeneous polynomial of degree 8), we arrive at the desired statement.
\end{proof}

\begin{lemma} \label{lemma p^2-q^2}
The following identity holds true:
\beq
p^2(x)-D q^2(x)= \frac{1}{16}\big({\rm tr} (f'(x))^2\big)^2-\frac{1}{16}\big({\rm tr} (f'(x)g'(x))\big)^2=\det f'(x).
\eeq
\end{lemma}
\begin{proof}
We transform its left-hand side of this identity as follows, using ${\rm tr} (f')^2={\rm tr} (g')^2$:
\begin{eqnarray}
\lefteqn{\Big({\rm tr} (f')^2-{\rm tr} (f'g')\Big)\Big({\rm tr} (f')^2+{\rm tr} (f'g')\Big)} \nonumber \\
& = & \Big(\frac{1}{2}{\rm tr} (f')^2+\frac{1}{2}{\rm tr} (g')^2-{\rm tr} (f'g')\Big)\Big(\frac{1}{2}{\rm tr} (f')^2+\frac{1}{2}{\rm tr} (g')^2+{\rm tr} (f'g')\Big) \nonumber \\
& = & \frac{1}{4}{\rm tr} (f'-g')^2\ {\rm tr} (f'+g')^2 \label{aux1}
\end{eqnarray}
Due to \eqref{char pol f'-g'}, \eqref{char pol f'+g'}, we find that the right-hand side of \eqref{aux1} is equal to the coefficient by $\lambda^4$ in
$$
\det\big(\lambda I- (f'-g'))(\lambda I- (f'+g')\big)=\det\big((\lambda I-f')^2-(g')^2\big)=\lambda^4\det(\lambda I-2f').
$$
(at the last step we used \eqref{f'f'=g'g'}). The latter coefficient is equal to $\det(2f')=16\det f'$, which finishes the proof.
\end{proof}

\begin{lemma} \label{lemma grads pq}
The following identity holds true:
\beq \label{grads pq}
\nabla p(x)=A \nabla q(x).
\eeq
As a consequence, both quadratic polynomials $p(x)$ and $q(x)$ satisfy the second-order differential equations \eqref{d2H}:
\beq  \label{d2 pq}
A(\nabla^2 p)=(\nabla^2 p) A^{\rm T}, \quad A(\nabla^2 q)=(\nabla^2 q) A^{\rm T}.
\eeq
\end{lemma}
\begin{proof}
We use the characteristic property \eqref{Af'} of the matrix $f'$ which in components reads:
\beq\label{Af' comp}
\sum_{j}(A^{\rm T})_{ij}\frac{\partial f_j}{\partial x_m}=\sum_{j}\frac{\partial f_i}{\partial x_j}(A^{\rm T})_{jm}\quad \forall i,m,
\eeq 
as well as its derivative with respect to $x_\ell$: 
\beq\label{Af'' comp}
\sum_{j}(A^{\rm T})_{ij}\frac{\partial^2 f_j}{\partial x_m\partial x_\ell}=\sum_{j}\frac{\partial^2 f_i}{\partial x_j\partial x_\ell}(A^{\rm T})_{jm}\quad \forall i,m,\ell.
\eeq 

We compute the components of $\nabla q$:
\begin{equation}\label{grad q comp}
\frac{\partial q}{\partial x_\ell}=
\frac{1}{4D}\frac{\partial\ {\rm tr}(A^{\rm T}(f')^2)}{\partial x_\ell}=
\frac{1}{4D}\sum_{i,j,m}(A^{\rm T})_{ij}\frac{\partial^2 f_j}{\partial x_m\partial x_\ell}\frac{\partial f_m}{\partial x_i}
+\frac{1}{4D}\sum_{i,j,m}(A^{\rm T})_{ij}\frac{\partial f_j}{\partial x_m}\frac{\partial^2 f_m}{\partial x_i\partial x_\ell}.
\end{equation}
The contribution to
\beq \label{A grad q comp}
(A\nabla q)_k=\sum_\ell (A^{\rm T})_{\ell k}\frac{\partial q}{\partial x_\ell}
\eeq
of the first sum in \eqref{grad q comp} is
\begin{eqnarray}
\lefteqn{\frac{1}{4D}\sum_\ell (A^{\rm T})_{\ell k} \sum_{i,m}\left(\sum_j (A^{\rm T})_{ij}\frac{\partial^2 f_j}{\partial x_m\partial x_\ell}\right)\frac{\partial f_m}{\partial x_i} \qquad {\rm (use\; \eqref{Af'' comp})}}\nonumber \\
& = & \frac{1}{4D}\sum_\ell (A^{\rm T})_{\ell k} \sum_{i,m}\left(\sum_j\frac{\partial^2 f_i}{\partial x_j\partial x_m} (A^{\rm T})_{j\ell}\right)\frac{\partial f_m}{\partial x_i}\nonumber\\
& = & \frac{1}{4D}\sum_{i,j,m} \frac{\partial^2 f_i}{\partial x_j\partial x_m}\frac{\partial f_m}{\partial x_i}\sum_\ell(A^{\rm T})_{\ell k}(A^{\rm T})_{j\ell}
\qquad {\rm (use}\; (A^{\rm T})^2=D I)
\nonumber\\
& = & \frac{1}{4}\sum_{i,j,m} \frac{\partial^2 f_i}{\partial x_j\partial x_m}\frac{\partial f_m}{\partial x_i}\delta_{jk} \; = \; 
\frac{1}{4}\sum_{i,m} \frac{\partial^2 f_i}{\partial x_k\partial x_m}\frac{\partial f_m}{\partial x_i}.
\end{eqnarray}
Similarly, the contribution of the second sum in \eqref{grad q comp} to \eqref{A grad q comp} is
\begin{eqnarray}
\lefteqn{\frac{1}{4D}\sum_\ell (A^{\rm T})_{\ell k} \sum_{i,m}\left(\sum_j (A^{\rm T})_{ij}\frac{\partial f_j}{\partial x_m}\right)\frac{\partial^2 f_m}{\partial x_i\partial x_\ell} \qquad {\rm (use\; \eqref{Af' comp})}}\nonumber \\
& = & \frac{1}{4D}\sum_\ell (A^{\rm T})_{\ell k} \sum_{i,m}\left(\sum_j \frac{\partial f_i}{\partial x_j}(A^{\rm T})_{jm}\right)\frac{\partial^2 f_m}{\partial x_i\partial x_\ell}\nonumber\\
& = & \frac{1}{4D}\sum_\ell (A^{\rm T})_{\ell k} \sum_{i,j} \frac{\partial f_i}{\partial x_j}\left(\sum_m (A^{\rm T})_{jm}\frac{\partial^2 f_m}{\partial x_i\partial x_\ell}\right) \qquad {\rm (use\; \eqref{Af'' comp})}\nonumber\\
& = & \frac{1}{4D}\sum_\ell (A^{\rm T})_{\ell k} \sum_{i,j} \frac{\partial f_i}{\partial x_j}\left(\sum_m \frac{\partial^2 f_j}{\partial x_i\partial x_m}
(A^{\rm T})_{m\ell}\right) \nonumber\\
& = & \frac{1}{4D}\sum_{i,j,m} \frac{\partial f_i}{\partial x_j} \frac{\partial^2 f_j}{\partial x_i\partial x_m}\sum_\ell(A^{\rm T})_{\ell k}(A^{\rm T})_{m\ell}
\qquad {\rm (use}\; (A^{\rm T})^2=D I)
\nonumber\\
& = & \frac{1}{4}\sum_{i,j,m} \frac{\partial f_i}{\partial x_j} \frac{\partial^2 f_j}{\partial x_i\partial x_m}\delta_{km} \; = \; 
\frac{1}{4}\sum_{i,j} \frac{\partial f_i}{\partial x_j} \frac{\partial^2 f_j}{\partial x_i\partial x_k} .
\end{eqnarray}
Collecting all the results, we find:
$$
(A\nabla q)_k=\frac{1}{4}\sum_{i,m} \frac{\partial^2 f_i}{\partial x_k\partial x_m}\frac{\partial f_m}{\partial x_i}+
\frac{1}{4}\sum_{i,j} \frac{\partial f_i}{\partial x_j} \frac{\partial^2 f_j}{\partial x_i\partial x_k}=\frac{1}{4}\frac{\partial\ {\rm tr}((f')^2)}{\partial x_k}=(\nabla p)_k,
$$
which finishes the proof.
\end{proof}

\section{Commutativity of maps}
\label{sect commut}

\begin{theorem}
The maps
\bea
&& \Phi_{f}: x \mapsto \t x = 
\left( I - \ep f'(x) \right)^{-1} x =
\left( I + \ep f'(\t x) \right) x, \label{eq: Phi1} \\
&& \Phi_{g}: x \mapsto \widehat x = 
\left( I - \ep g'(x) \right)^{-1} x =
\left( I + \ep g'(\widehat x) \right) x, \label{eq: Phi2} 
\eea
commute:  $\Phi_{f} \circ \Phi_{g}=\Phi_{g} \circ \Phi_{f}$.
\end{theorem}

\begin{proof}
We have:
\begin{equation}\label{eq: hat tilde}
\left(\Phi_{g} \circ \Phi_{f}\right)(x)=\left( I - \ep g'(\t x) \right)^{-1} \left( I + \ep f'(\t x) \right) x,
\end{equation}
and
\begin{equation}\label{eq: tilde hat}
\left(\Phi_{f} \circ \Phi_{g}\right)(x)=\left( I - \ep f'(\widehat x) \right)^{-1} \left( I + \ep g'(\widehat x) \right) x.
\end{equation}
We prove the following matrix equation:
\beq \label{det}
\left( I - \ep g'(\t x) \right)^{-1} \left( I + \ep f'(\t x) \right) =\left( I - \ep f'(\widehat x) \right)^{-1} \left( I + \ep g'(\widehat x) \right),
\eeq
which is stronger than the vector equation $\left(\Phi_{f}\circ\Phi_{g}\right)(x)=\left(\Phi_{g}\circ\Phi_{f}\right)(x)$ expressing commutativity. Equation \eqref{det} is equivalent to
\beq \label{det 1}
\left( I - \ep f'(\widehat x) \right) \left( I - \ep g'(\t x) \right)^{-1} =
  \left( I + \ep g'(\widehat x) \right)\left( I + \ep f'(\t x) \right)^{-1}.
\eeq
From \eqref{f'f'=g'g'} we find:
\[
\left( I - \ep g'(\t x) \right)^{-1}=\left( I + \ep g'(\t x) \right)\left( I - \ep^2 (f'(\t x))^2 \right)^{-1},
\]
\[
\left( I + \ep f'(\t x) \right)^{-1}=\left( I - \ep f'(\t x) \right)\left( I - \ep^2(f'(\t x))^2 \right)^{-1}.
\]
With this at hand, equation \eqref{det 1} is equivalent to
\[
\left( I - \ep f'(\widehat x) \right)\left( I + \ep g'(\t x) \right)=
\left( I + \ep g'(\widehat x) \right)\left( I - \ep f'(\t x) \right).
\]
Here the quadratic in $\epsilon$ terms cancel by virtue of \eqref{f' vs g'} and \eqref{Af'}:
\[
f'(\widehat x)g'(\t x) = \alpha^{-1}f'(\widehat x)A^{\rm T}f'(\t x) =  \alpha^{-1} A^{\rm T}f'(\widehat x)f'(\t x)  =  g'(\widehat x)f'(\t x),
\]
so that we are left with the terms linear in $\epsilon$:
\beq \label{aux}
-  f'(\widehat x)+ g'(\t x)= g'(\widehat x)- f'(\t x).
\eeq
Since the tensors $f''$, $g''$ are constant, we have:
\[
f'(\widehat x)=f'(x)+f''(\widehat x-x)=f'(x)+2\ep f''\left(I-\ep g'(x)\right)^{-1}g(x),
\]
\[
g'(\widehat x)=g'(x)+g''(\widehat x-x)=g'(x)+2\ep g''\left(I-\ep g'(x)\right)^{-1}g(x),
\]
\[
f'(\t x)=f'(x)+f''(\t x-x)=f'(x)+2\ep f''\left(I-\ep f'(x)\right)^{-1}f(x),
\]
\[
g'(\t x)=g'(x)+g''(\t x-x)=g'(x)+2\ep g''\left(I-\ep f'(x)\right)^{-1}f(x).
\]
Thus, equation \eqref{aux} is equivalent to 
\begin{align} \label{aux2}
& f''\left(I-\ep g'(x)\right)^{-1}g(x)+g''\left(I-\ep g'(x)\right)^{-1}g(x) = \nonumber\\
& \qquad f''\left(I-\ep f'(x)\right)^{-1}f(x)+g''\left(I-\ep f'(x)\right)^{-1}f(x).
\end{align}
At this point, we use the following statement.

\begin{lemma} \label{lemma g''}
For any vector $v\in\mathbb C^4$ we have:
\beq \label{g''}
 g''(x)v=\alpha^{-1}f''(x)(A^{\rm T}v), \quad  f''(x)v=\alpha^{-1}g''(x)(A^{\rm T}v).
\eeq
\end{lemma}

\noindent
We compute the matrices on the left-hand side of \eqref{aux2} with the help of \eqref{g''}, \eqref{f vs g}, \eqref{f' vs g'}:
\begin{eqnarray*}
f''\left(I-\ep g'(x)\right)^{-1}g(x) & = & f''\left(I-\ep^2(f'(x))^2\right)^{-1}\left(g(x)+\ep g'(x)g(x)\right),\\
g''\left(I-\ep g'(x)\right)^{-1}g(x) & = & \alpha^{-1}f''\left(I-\ep^2(f'(x))^2\right)^{-1}A^{\rm T}\left(g(x)+\ep g'(x)g(x)\right)\\
                                                              & = & f''\left(I-\ep^2 (f'(x))^2\right)^{-1}\left(f(x)+\ep f'(x)g(x)\right),
\end{eqnarray*}
and similarly    
\begin{eqnarray*}
f''\left(I-\ep f'(x)\right)^{-1}f(x) & = & f''\left(I-\ep^2 (f'(x))^2\right)^{-1}\left(f(x)+\ep f'(x)f(x)\right)\\
g''\left(I-\ep f'(x)\right)^{-1}f(x) & = & \alpha^{-1}f''\left(I-\ep^2 (f'(x))^2\right)^{-1}A^{\rm T}\left(f(x)+\ep f'(x)f(x)\right)\\
                                                              & = & f''\left(I-\ep^2 (f'(x))^2\right)^{-1}\left(g(x)+\ep g'(x)f(x)\right).
\end{eqnarray*}                                                          
Collecting all the results and using \eqref{vf comm} and \eqref{f'f=g'g}, we see that the proof is complete. 
\end{proof}

{\em Proof of Lemma \ref{lemma g''}.} The identities in question are equivalent to
\beq \label{Af''}
 A^{\rm T}(f''(x)v)=f''(x)(A^{\rm T}v), \quad  A^{\rm T}(g''(x)v)=g''(x)(A^{\rm T}v).
\eeq
(Actually, both tensors $f''$ and $g''$ are constant, i.e., do not depend on $x$.) To prove the latter identities, we start with equation \eqref{Af'} written in components:
\[
\sum_k (A^{\rm T})_{ik}\frac{\partial f_k}{\partial x_\ell}=\sum_k \frac{\partial f_i}{\partial x_k}(A^{\rm T})_{k\ell}.
\]
Differentiating with respect to $x_j$, we get:
\[
\sum_k (A^{\rm T})_{ik}\frac{\partial^2 f_k}{\partial x_j\partial x_\ell}=\sum_k \frac{\partial f_i}{\partial x_j\partial x_k}(A^{\rm T})_{k\ell}.
\]
Hence,
\[
\sum_{k,\ell} (A^{\rm T})_{ik}\frac{\partial^2 f_k}{\partial x_j\partial x_\ell}v_\ell=\sum_{k,\ell} \frac{\partial f_i}{\partial x_j\partial x_k}(A^{\rm T})_{k\ell}v_\ell,
\]
which is nothing but the $(i,j)$ entry of the matrix identity \eqref{Af''}. \qed

\section{Integrals of motion}
\label{sect integrals}

\begin{theorem} \label{th integrals}
The maps $\Phi_{f}$ and $\Phi_{g}$ share two functionally independent conserved quantities 
\begin{equation}\label{tH}
\t H(x,\epsilon)=\epsilon^{-1}x^{\rm T}J\ \t x=\epsilon^{-1}x^{\rm T}J\left(I-\ep f'(x)\right)^{-1}x
\end{equation}
and 
\begin{equation}\label{tK}
\t K(x,\epsilon)=\epsilon^{-1}\alpha x^{\rm T}J\ \widehat x=\epsilon^{-1}\alpha x^{\rm T}J\left(I-\ep g'(x)\right)^{-1}x.
\end{equation}
\end{theorem}
Before proving this theorem, we observe different expressions for these functions. Expanding \eqref{tH} in power series with respect to $\epsilon$, we see:
$$
\t H(x, \epsilon) = \sum_{k=0}^\infty \epsilon^{k-1} x^{\rm T}J(f'(x))^kx=-\sum_{k=0}^\infty \epsilon^{k-1}x^{\rm T} \Big(\nabla^2 H(x) J\nabla^2 H(x) \cdots J\nabla^2 H(x)\Big) x.
$$
The matrix in the parentheses (involving $k$ times $\nabla^2 H(x)$ and $k-1$ times $J$) is symmetric if $k$ is odd, and skew-symmetric if $k$ is even. Therefore, all terms with even $k$ vanish, and we arrive at
$$
\t H(x, \epsilon) = \sum_{k=0}^\infty \epsilon^{2k} x^{\rm T}J(f'(x))^{2k+1}x,
$$
or, finally,
\begin{equation} \label{tH alt}
\t H(x, \epsilon) =x^{\rm T} J\left(I-\epsilon^2(f'(x))^2\right)^{-1} f'(x) x =
2x^{\rm T} J\left(I-\epsilon^2(f'(x))^2\right)^{-1} f(x) .
\end{equation}
At the last step we used that $f(x)$ is homogeneous of degree 2, so that $f'(x)x=2f(x)$.  Of course, analogous expressions hold true for the function $\t K(x,\epsilon)$:
\begin{equation} \label{tK alt}
\t K(x, \epsilon) =\alpha x^{\rm T} J\left(I-\epsilon^2(g'(x))^2\right)^{-1} g'(x) x =
2\alpha x^{\rm T} J\left(I-\epsilon^2(g'(x))^2\right)^{-1} g(x) .
\end{equation}
Formulas \eqref{tH alt}, \eqref{tK alt} also clearly display the asymptotics
$$
\t H(x, \epsilon)=2x^{\rm T} Jf(x) +O(\epsilon^2)=-2x^{\rm T} \nabla H(x)+O(\epsilon^2)=-6H(x)+O(\epsilon^2),
$$
and analogously 
$$
\t K(x, \epsilon)=2\alpha x^{\rm T} Jg(x) +O(\epsilon^2)=-2x^{\rm T} \nabla K(x)+O(\epsilon^2)=-6K(x)+O(\epsilon^2).
$$

{\em Proof of Theorem \ref{th integrals}.}
We first show that $\t H(x,\epsilon)$ is an integral of motion of the map $\Phi_f$ (this is a result from \cite{CMOQ1}, which holds true for arbitrary Hamiltonian vector fields). For this goal, we compute with the help of \eqref{eq: Phi1}:
\begin{eqnarray*}
\t H(\t x,\epsilon) & = & \epsilon^{-1}\t x^{\rm T}J\left(I-\ep f'(\t x)\right)^{-1}\t x  \\
 & = & \epsilon^{-1}x^{\rm T}\left(I+\epsilon f'(\t x)\right)^{\rm T}J\left(I-\ep f'(\t x)\right)^{-1}\left(I-\epsilon f'(x)\right)^{-1} x .
\end{eqnarray*}
Taking into account \eqref{Ham f'} in the form $(f'(\t x))^{\rm T} J=-Jf'(\t x)$, we arrive at
\begin{equation}
\t H(\t x,\epsilon)\; =\; \epsilon^{-1}x^{\rm T}J(I-\epsilon f'(\t x))\left(I-\ep f'(\t x)\right)^{-1}\left(I-\epsilon f'(x)\right)^{-1} x \; = \; \t H(x,\epsilon).
\end{equation}

Next, we show that $\t K(x,\epsilon)$ also is an integral of motion of the map $\Phi_f$. For this goal, we first compute, based on \eqref{tK alt}:
\begin{eqnarray*}
\t K(\t x,\epsilon) & = & \alpha \t x^{\rm T}J\left(I-\ep^2 (g'(\t x))^2\right)^{-1}g'(\t x)\t x  \\
 & = &  \alpha x^{\rm T}\left(I+\epsilon f'(\t x)\right)^{\rm T}J\left(I-\ep^2 (g'(\t x))^2\right)^{-1}g'(\t x)\left(I+\epsilon f'(\t x)\right) x \\
  & = &  \alpha x^{\rm T}J\left(I-\epsilon f'(\t x)\right)\left(I-\ep^2 (g'(\t x))^2\right)^{-1}g'(\t x)\left(I+\epsilon f'(\t x)\right) x .
\end{eqnarray*}
By virtue of \eqref{f'g'=g'f'} and \eqref{f'f'=g'g'} we arrive at
\begin{equation}\label{aux3}
\t K(\t x,\epsilon)=  \alpha x^{\rm T}Jg'(\t x) x .
\end{equation}
Now we compute as in the previous section:
\begin{eqnarray}\label{aux4}
g'(\t x) & = & g'(x)+g''(\t x-x)=g'(x)+2\ep g''\left(I-\ep f'(x)\right)^{-1}f(x)\nonumber\\
 & = & g'(x)+2\ep g''\left(I-\ep^2 (f'(x))^2\right)^{-1}\left(I+\ep f'(x)\right)f(x).
\end{eqnarray}
We will show that the contribution to \eqref{aux3} of the terms in \eqref{aux4} with odd powers of $\epsilon$ vanishes:
\begin{eqnarray}\label{aux5}
\alpha x^{\rm T}J g''\left(I-\ep^2 (f'(x))^2\right)^{-1}f(x)x=0.
\end{eqnarray}
For this, we use the fact that for an arbitrary vector $v\in\mathbb C^4$  we have:
\begin{equation} \label{aux7}
g''vx=g'(x)v.
\end{equation}
Indeed,  due to homogeneity of $g'(x)$,
$$
(g''vx)_i=\sum_{j,k=1}^4 \frac{\partial^2 g_i}{\partial x_j\partial x_k}v_jx_k=\sum_{j=1}^4\frac{\partial g_i}{\partial x_j}v_j=(g'(x)v)_i.
$$
Due to \eqref{aux7}, the left-hand side of \eqref{aux5} is equal  to
\begin{eqnarray*}
\lefteqn{ \alpha x^{\rm T}J g'(x)\left(I-\ep^2 (f'(x))^2\right)^{-1}f(x) } \\
& = & -x^{\rm T}\nabla^2 K(x)\left(I-\ep^2 (f'(x))^2\right)^{-1}f(x)\\
& = & -2(\nabla K(x))^{\rm T}\left(I-\ep^2 (f'(x))^2\right)^{-1}f(x) \\
& = & -2(\nabla H(x))^{\rm T}\Big(A^{\rm T}\left(I-\ep^2 (f'(x))^2\right)^{-1} J\Big)\nabla H(x).
\end{eqnarray*}
One easily sees with the help of \eqref{Ham f'}, \eqref{Af'} and \eqref{cond A} that the matrix $A^{\rm T}\left(I-\ep^2 (f'(x))^2\right)^{-1} J$ is skew-symmetric, which finishes the proof of \eqref{aux5}. With this result, \eqref{aux3} turns into 
\[
\t K(\t x,\epsilon)= \alpha x^{\rm T}Jg'(x) x +2\alpha \epsilon^2g''\left(I-\ep^2 (f'(x))^2\right)^{-1} f'(x)f(x)x.
\]
By virtue of \eqref{f'f'=g'g'}, \eqref{f'f=g'g} and  \eqref{aux7}, we put the latter formula as
\begin{eqnarray*}
\t K(\t x,\epsilon) & = &  \alpha x^{\rm T}Jg'(x) x +2\epsilon^2 \alpha x^{\rm T}Jg''\left(I-\ep^2 (g'(x))^2\right)^{-1} g'(x)g(x)x \\
 & = & 2 \alpha x^{\rm T}Jg(x) +2\epsilon^2 \alpha x^{\rm T}Jg'(x)\left(I-\ep^2 (g'(x))^2\right)^{-1} g'(x)g(x) \\
 & = & 2 \alpha x^{\rm T} J \left(I-\ep^2 (g'(x))^2\right)^{-1} g(x) \; = \; \t K(x,\epsilon).
\end{eqnarray*}
This finishes the proof of the theorem.
\qed

\section{Invariant Poisson structure}
\label{sect symplectic}

\begin{theorem}
Both maps $\Phi_f$ and $\Phi_g$ are Poisson with respect to the brackets with the Poisson tensor $\Pi(x)$ given by
\begin{eqnarray} \label{Pi short}
\Pi(x)  & = & J- \epsilon^2 (f'(x))^2J \label{Pi short} \\
 & = & \big(1-\epsilon^2p(x)\big)J-\epsilon^2 q(x) A^{\rm T} J, \label{Pi medium}
\end{eqnarray}
where $p(x)$ and $q(x)$ are quadratic polynomials given in \eqref{p}, \eqref{q}.
\end{theorem}
\begin{proof}
First, we prove that
\beq \label{Pois prop}
d\Phi_f(x)\Pi(x) (d\Phi_f(x))^{\rm T} =\Pi(\t x).
\eeq
With the expression \eqref{Jac} for $d\Phi_f(x)$, \eqref{Pois prop} turns into
$$
\big(I+\epsilon f'(\t x)\big)\Pi(x)\big(I+\epsilon f'(\t x)\big)^{\rm T}=\big(I-\epsilon f'(x)\big)\Pi(\t x)\big(I-\epsilon f'(x)\big)^{\rm T}.
$$
Multiplying from the right by $J$ and taking into account \eqref{Ham f'}, we arrive at:
$$
\big(I+\epsilon f'(\t x)\big)\Pi(x)J\big(I-\epsilon f'(\t x)\big)=
\big(I-\epsilon f'(x)\big)\Pi(\t x)J\big(I+\epsilon f'(x) \big).
$$
According to Lemma \ref{lemma f'f'}, the matrix $\Pi(x)J$ is a linear combination of $I$ and $A^{\rm T}$, therefore, by virtue of \eqref{Af'}, it commutes with $f'(\t x)$ (actually, with $f'$ evaluated at any point). Thus, the latter equation is equivalent to 
$$
\big(I-\epsilon (f'(\t x))^2\big)\Pi(x)J=\Pi(\t x)J\big(I-\epsilon (f'(x))^2\big),
$$
which is obviously true due to \eqref{Pi short} .

It remains to prove that $\Pi$ is indeed a Poisson tensor. For this, one has to verify the Jacobi identity
\beq \label{Jacobi id}
\{x_i,\{x_j,x_k\}\}+\{x_j,\{x_k,x_i\}\}+\{x_k,\{x_i,x_j\}\}=0
\eeq
for the four different triples of indices $\{i,j,k\}$ from $\{1,2,3,4\}$. A straightforward computation based on the expression \eqref{Pi medium} for $\Pi$ shows that the left-hand sides of the expressions \eqref{Jacobi id} are polynomials of order 2 in $\epsilon^2$, with the coefficients by $\epsilon^2$ being the corresponding components of the vector $\nabla p(x)-A\nabla q(x)$, and the coefficients by $\epsilon^4$ being linear combinations of the latter. A reference to Lemma \ref{lemma grads pq} finishes the proof.
\end{proof}


\section{Differential equations for the conserved quantities of maps $\Phi_f$, $\Phi_g$}
\label{sect diff eqs}

\begin{theorem} \label{th diff eqs}
The rational functions $\t H(x,\epsilon)$, $\t K(x,\epsilon)$ are related by
\begin{equation}\label{grad tK}
\nabla \t K(x,\epsilon) =A \nabla \t H(x,\epsilon).
\end{equation}
As a consequence, they satisfy the same second order differential equations \eqref{comp1}--\eqref{comp5} as the cubic polynomials $H(x)$, $K(x)$.
\end{theorem}
\begin{proof}
We start the proof with the derivation of the following formula for $\t H(x,\epsilon)$:
\begin{equation} \label{tH thru H K}
\t H(x,\epsilon) = -6\dfrac{ (1-\epsilon^2 p(x))H(x)+\epsilon^2 q(x) K(x)}{(1-\epsilon^2 p(x))^2-\epsilon^4 \alpha^2 q^2(x) }.
\end{equation}
For this aim, we observe the following formula:
\begin{equation} \label{Pi inv short}
(\Pi(x))^{-1}  = -\frac{(1-\epsilon^2 p(x)) J + \epsilon^2 q(x) AJ}{(1-\epsilon^2 p(x))^2-\epsilon^4 \alpha^2 q^2(x) }.
\end{equation}
This is checked by a straightforward multiplication of expressions \eqref{Pi short}, \eqref{Pi inv short}. We note that, by Lemma \ref{lemma p^2-q^2}, 
the denominator in \eqref{Pi inv short} is equal to
$$
1-\frac{\epsilon^2}{2} {\rm tr} (f')^2+\epsilon^4 \det f'=\det(I-\epsilon f').
$$
Now we derive from \eqref{tH alt}, \eqref{Pi short}, \eqref{Pi medium}:
\begin{eqnarray*}
\t H(x,\epsilon) & = & 2x^{\rm T} J\left(I-\epsilon^2(f'(x))^2\right)^{-1} f(x) \\
& = & -2x^{\rm T} (\Pi(x))^{-1}f(x)\\
& = & 2\dfrac{ (1-\epsilon^2 p(x))x^{\rm T} J f(x)+ \epsilon^2 q(x) x^{\rm T} AJf(x)}{\det(I-\epsilon f'(x))}\\
& = & -2\dfrac{ (1-\epsilon^2 p(x))x^{\rm T} \nabla H(x)+ \epsilon^2 q(x) x^{\rm T} \nabla K(x)}{\det(I-\epsilon f'(x))}\\
& = & -6\dfrac{ (1-\epsilon^2 p(x))H(x)+\epsilon^2 q(x) K(x)}{\det(I-\epsilon f'(x))}, 
\end{eqnarray*}
which is formula \eqref{tH thru H K}. Interchanging the roles of $H$ and $K$, that is, replacing $(H,K)$ by $(K, \alpha^2 H)$, we find:
\begin{equation} \label{tK thru H K}
\t K(x,\epsilon) = -6\dfrac{ (1-\epsilon^2 p(x))K(x)+\epsilon^2 \alpha^2 q(x) H(x)}{(1-\epsilon^2 p(x))^2-\epsilon^4 \alpha^2 q^2(x) }.
\end{equation}
We write formulas \eqref{tH thru H K}, \eqref{tK thru H K} as
\begin{eqnarray} 
\t H(x,\epsilon) & = & r(x)H(x)+s(x) K(x), \label{tH thru H K short}\\
\t K(x,\epsilon) & = & r(x)K(x)+\alpha^2 s(x) H(x). \label{tK thru H K short}
\end{eqnarray}
We will prove the following relation for the coefficients $r(x)$, $s(x)$:
\beq \label{grads rs}
\nabla r(x)=A\nabla s(x).
\eeq
Then \eqref{grad tK} will be an immediate corollary of \eqref{tH thru H K short}, \eqref{tK thru H K short}, combined with \eqref{grad K} and \eqref{grads rs}. Indeed, we have:
\begin{eqnarray*}
\nabla \t H & = & r\nabla H+H\nabla r +s \nabla K+K\nabla s \\
& = &  r\nabla H+HA\nabla s +sA \nabla H+K\nabla s,\\
\nabla \t K & = &  r\nabla K+K\nabla r +\alpha^2 s \nabla H+\alpha^2 H\nabla s \\
& = &  rA\nabla H+KA\nabla s +\alpha^2 s \nabla H+\alpha^2 H\nabla s,
\end{eqnarray*}
and \eqref{grad tK} follows directly by virtue of $A^2=\alpha^2 I$.

Thus, it remains to prove \eqref{grads rs}. We have:
\begin{eqnarray}
-\tfrac{1}{3}r(x) & = & \dfrac{2(1-\epsilon^2 p(x))}{(1-\epsilon^2 p(x))^2-\epsilon^4\alpha^2 q^2(x) } \nonumber\\
 & = & \dfrac{1}{1-\epsilon^2 p(x)-\epsilon^2\alpha q(x) } +\dfrac{1}{1-\epsilon^2 p(x)+\epsilon^2\alpha q(x) },\\
 -\tfrac{1}{3}s(x) & = & \dfrac{2\epsilon^2 q(x)}{(1-\epsilon^2 p(x))^2-\epsilon^4\alpha^2 q^2(x) } \nonumber\\
 & = & \dfrac{\alpha^{-1}}{1-\epsilon^2 p(x)-\epsilon^2\alpha q(x) } -\dfrac{\alpha^{-1}}{1-\epsilon^2 p(x)+\epsilon^2\alpha q(x) }.
\end{eqnarray}
Therefore, by virtue of \eqref{grads pq}:
\begin{eqnarray}
-\tfrac{1}{3}\epsilon^{-2}\nabla r(x) & = & \dfrac{\nabla p(x)+\alpha \nabla q(x)}{\big(1-\epsilon^2 p(x)-\epsilon^2\alpha q(x)\big)^2 } 
                                           +\dfrac{\nabla p(x)-\alpha \nabla q(x)}{\big(1-\epsilon^2 p(x)+\epsilon^2\alpha q(x)\big)^2 }\nonumber\\
                 & = & \dfrac{(A+\alpha I) \nabla q(x)}{\big(1-\epsilon^2 p(x)-\epsilon^2\alpha q(x)\big)^2 } 
                                           +\dfrac{(A-\alpha I) \nabla q(x)}{\big(1-\epsilon^2 p(x)+\epsilon^2\alpha q(x) \big)^2},\nonumber\\                          
 -\tfrac{1}{3}\epsilon^{-2}\nabla s(x) & = & \dfrac{\alpha^{-1}(\nabla p(x)+\alpha \nabla q(x))}{\big(1-\epsilon^2 p(x)-\epsilon^2\alpha q(x)\big)^2 } 
                                           -\dfrac{\alpha^{-1}(\nabla p(x)-\alpha \nabla q(x))}{\big(1-\epsilon^2 p(x)+\epsilon^2\alpha q(x)\big)^2 }\nonumber\\
                 & = & \dfrac{\alpha^{-1}(A+\alpha I) \nabla q(x)}{\big(1-\epsilon^2 p(x)-\epsilon^2\alpha q(x)\big)^2 } 
                                           -\dfrac{\alpha^{-1}(A-\alpha I) \nabla q(x)}{\big(1-\epsilon^2 p(x)+\epsilon^2\alpha q(x) \big)^2}.\nonumber
\end{eqnarray}
Now \eqref{grads rs} follows by virtue of
\begin{align*}
& \alpha^{-1}A(A+\alpha I)=A+\alpha ^{-1}A^2=A+\alpha I, \\
& \alpha^{-1}A(A-\alpha I)=-A+\alpha^{-1}A^2=-(A-\alpha I).
\end{align*}
 This finishes the proof. 
\end{proof}


\section{Conclusions}

Completely integrable Hamiltonian systems lying at the basis of our constructions, seem to be worth studying on their own. In particular, their invariant surfaces are Abelian varieties appearing as intersections of two cubic hypersurfaces in the 4-dimensional space. Algebraic geometry of such surfaces does not seem to be elaborated very well in the existing literature. Still more interesting and intriguing are the algebraic-geometric aspects of the commuting pairs of integrable maps introduced here. Moreover, a fair part of our constructions seems to be generalizable into higher dimensions. All this will be the subject of our future research. 

\section{Acknowledgements}

This research is supported by the DFG Collaborative Research Center TRR 109 ``Discretization in Geometry and Dynamics''.


\end{document}